\documentclass[11pt,english]{article}

\usepackage{amsmath,amsthm,amssymb,amsfonts}
\usepackage{xspace}
\usepackage{url}
\usepackage{hyperref}
\usepackage[dvipsnames]{xcolor}
\usepackage{thmtools}
\usepackage{thm-restate}
\usepackage{amsmath}

\hypersetup{colorlinks=true,allcolors=blue}
\usepackage[letterpaper, left=1in, right=1in, top=0.9in, bottom=0.9in]{geometry}
\DeclareMathOperator{\Res}{Res}
\DeclareMathOperator{\Pre}{Pre}
\DeclareMathOperator{\Alg}{Alg}
\usepackage{multirow}
\usepackage{longtable}
\usepackage{booktabs}
\usepackage{thm-restate}
\usepackage{graphicx}
\usepackage{subcaption}
\usepackage{enumitem}

\usepackage[capitalise,noabbrev]{cleveref}
\usepackage{comment}

\theoremstyle{plain}
\newtheorem{theorem}{Theorem}[section]
\newtheorem{lemma}[theorem]{Lemma}
\newtheorem{claim}[theorem]{Claim}

\usepackage{parskip}

\theoremstyle{definition}
\newtheorem{definition}[theorem]{Definition}

\usepackage[ruled,vlined]{algorithm2e}

\usepackage{mathtools}

\theoremstyle{remark}

\newtheorem*{remark*}{Remark}

\usepackage[obeyFinal,textsize=footnotesize]{todonotes}

\title{Online Preemptive Matching Revisited}
\author{
Peter Kiss\thanks{This research was funded in whole or in part by the Austrian Science Fund
(FWF) 10.55776/ESP6088024}$^{1}$,
Mohammad Sharifi$^2$ \\
\\
$^{1}$University of Vienna\\
$^{2}$Sharif University of Technology
}
\date{}

\begin{document}
\maketitle
\begin{abstract}

We study the online preemptive matching problem, in which the edges of a graph arrive sequentially and the algorithm must maintain a matching by accepting or rejecting arriving edges and possibly discarding previously accepted ones. We prove a new upper bound of $0.5661$ on the competitive ratio achievable for the problem. This bound applies to arbitrary randomized algorithms, bipartite graphs and if we allow the algorithm to output a fractional solution.

Our result improves upon the strongest previously known upper bound of $2-\sqrt{2} \approx 0.585$, due to Huang et al. [SODA'19]. Previous hardness constructions relied on edge sequences described by vertex arrivals where each arriving vertex reveals its edges to yet unvaried vertices. Under such sequences, Huang et al. showed that there exists a non-preemptive online algorithm with competitive ratio $\sim0.567$ (or $2-\sqrt{2}$ for fractional solutions). Consequently, our hardness construction is the first result which shows hardness for instances where the optimal algorithm employs preemption.

\end{abstract}

\newpage

\section{Introduction}

We study the problem of maintaining a matching in an online graph under preemption (sometimes referred to as online matching with free disposal). The input is a simple graph $G = (V,E)$ whose edges arrive one by one in an adversarial order. Upon the arrival of an edge $e \in E$, the algorithm must decide whether to incorporate $e$ into its current matching or to discard it permanently. In contrast to the classical online setting, the algorithm is also allowed to \emph{preempt} previously selected edges, i.e., to remove edges arbitrarily from its current output. The objective is to maximize the size of the final matching. An algorithm is said to be \emph{$\alpha$-competitive} if, for every input sequence, the size of the matching it produces is at least an $\alpha$-fraction of the size of an optimal offline solution.

Since its introduction to the online algorithms literature in the seminal work of Karp, Vazirani, and Vazirani~\cite{KarpVV90}, the online matching problem has received extensive attention. This interest stems both from its appealing theoretical nature and from its practical relevance (see \cite{roth2023online}), most notably in the modeling of online advertisement markets. The problem in the absence of preemption is better understood. Gamlath et al.~\cite{GamlathKMSW19} showed that under adversarial edge arrivals, no algorithm can achieve a competitive ratio exceeding $1/2$, a bound attained by the simple greedy algorithm. Much of the literature has focused on models in which vertices arrive sequentially and reveal some or all of their incident edges. Depending on the precise model, strictly better than $1/2$-competitive algorithms are known~\cite{0002KTWZZ18,WangW15,GamlathKMSW19,KarpVV90,tang2024improved, huang2020fully} (for a more complete list see \cite{huang2024online}).

Given the importance of the problem, many works have explored relaxations of the purely online model. In these settings, the algorithm may receive additional information about the underlying graph~\cite{AntoniadisGKK23,ChooGL024,LiYR23}, or may be allowed to perform limited non-online modifications to its output~\cite{AngelopoulosDJ20,0004KS20,BernsteinHR19}. The preemptive variant of online matching, the focus of this paper, is both a natural relaxation of the classical model and a source of theoretical and practical applications. Notably, state-of-the-art hardness results for the semi-streaming matching problem rely on hardness constructions for preemptive online matching~\cite{Kapralov13,Kapralov21}.

The current state of the art for preemptive matching is (arguably) somewhat unsatisfactory. As in the purely online setting, no algorithm with competitive ratio strictly better than $1/2$ is known under adversarial edge arrivals. In contrast, existing hardness results for preemptive matching focus on vertex arrival models for which competitive ratios strictly exceeding $1/2$ are achievable even without preemption.

\paragraph{Existing Work.}
Define an \emph{online vertex-future edge arrival} sequence as follows. Vertices arrive sequentially, and upon the arrival of a vertex $v$, all edges between $v$ and vertices that have not yet arrived are revealed. At this point, the algorithm must decide whether and how to match $v$. This model is at least as hard as the fully online vertex arrival model (introduced by Huang et al. \cite{0002KTWZZ18}, see \cite{FullyOnlineIntegral, tang2024improved}), in which vertices have arrival times and deadlines (by which they must be matched) and only reveal edges to previously arrived neighbors (for a brief discussion on this, see Section~\ref{sec:ranking}). At the same time, it generalizes the classical online bipartite matching problem, in which vertices arrive on only one side of the bipartition.

Existing hardness results for preemptive matching rely on a key observation: under vertex-future edge arrival sequences, preemption is never beneficial. To see this, consider the first vertex $v$ whose incident edges are revealed, and suppose the algorithm matches $v$ to some neighbor $u$. By definition, no further edges incident to $v$ will arrive. Thus, the only possible reason to preempt the edge $(v,u)$ is to rematch $u$ later, but such a rematching cannot increase the size of the output.

Consequently, any hardness result that applies to purely online algorithms under vertex-future edge arrivals also holds for preemptive algorithms. Based on this observation, the classical hardness bound of $1 - 1/e \approx 0.632$ for online bipartite matching due to Karp et al.~\cite{KarpVV90} was the first known upper bound for online preemptive matching. Subsequent works by Epstein et al.~\cite{EpsteinLSW18Weaker} and Huang et al.~\cite{HuangPTTWZ19Stronger} extended these ideas to bipartite graphs with vertex arrivals on both sides of the bipartition, establishing upper bounds of $1/(1+\ln 2) \approx 0.59$ and $2-\sqrt{2} \approx 0.585$, respectively.

On the algorithmic side, Huang et al.~\cite{HuangPTTWZ19Stronger} showed that there exists a non-preemptive online algorithm with competitive ratio approximately $0.567$ for the fully online vertex arrival model (or $2-\sqrt{2}$ in the fractional setting). Although this model is not equivalent to vertex-future edge arrivals, it is straightforward to show (see Section~\ref{sec:ranking}) that these positive results extend to our formulation.

This leads to a notable gap in the literature. On the one hand, no algorithm is known to leverage preemption to surpass the greedy $1/2$ competitive ratio under adversarial edge arrivals. On the other hand, all existing hardness results circumvent the preemptive nature of the problem by relying on vertex-future edge arrivals, and therefore cannot rule out the existence of a better than $1/2$ competitive algorithm.

\subsection{Our Result}

Our main contribution is an improved upper bound for the online preemptive matching problem.

\begin{theorem}
\label{thm:main}
Any online preemptive matching algorithm has competitive ratio at most $0.5661$.
\end{theorem}

Theorem~\ref{thm:main} holds even if the algorithm is randomized or outputs a fractional matching, and if the input graph is bipartite. Note that its a folklore fact that for any randomized integral algorithm for the problem there exists a deterministic fractional one with the same competitive ratio. A key novelty of our approach is that we construct instances on which an optimal algorithm \emph{does} make use of preemption, thereby more faithfully capturing the inherent difficulty of the model than prior hardness constructions.

Our hard instance underlying Theorem~\ref{thm:main} is based on an edge arrival sequence that can be viewed as a special case of a \emph{$k$-vertex-future edge arrival} model. In this model, each vertex appears in $k$ distinct arrival phases, and upon each appearance it reveals a subset of its incident edges to vertices that have appeared fewer times. For our results, it suffices to consider $k = O(1)$; in fact, choosing $k = 2$ is already sufficient to improve upon the previous hardness bound of $2-\sqrt{2}$.

Finally, we note that Kapralov’s approximation hardness results for semi-streaming matching~\cite{Kapralov13,Kapralov21} rely on preemptive matching hardness constructions derived from vertex-future edge instances due to Karp et al.~\cite{KarpVV90} and Epstein et al.~\cite{EpsteinLSW18Weaker}. Since our constructions constitute only a mild relaxation of the model, we hope they may also serve as a useful tool for proving stronger hardness results in the semi-streaming setting.

\paragraph{Concurrent Work.} Concurrently and independently with our work, Assadi, Xiang and Jiang have listed a paper titled 'Semi-Streaming Matching in a Single Pass: A New Framework for Lower Bounds via Blueprints' (appearing at STOC26) on their website which might have implications for the preemptive model.

\paragraph{Related Work.} Preemptive matching is known to admit better then $1/2$-approximate algorithms for structured inputs. Jiang and Zhang \cite{JiangZ24}, Chiplunkar et al. \cite{DBLP:conf/esa/ChiplunkarTV15}, Tirodkar and Vishwanathan \cite{DBLP:conf/cocoon/TirodkarV17} and Buchbinder et al. \cite{DBLP:journals/algorithmica/BuchbinderST19} have studied the problem for acyclic graphs. Pashkovich and Snow \cite{DBLP:conf/waoa/PashkovichS25} have considered inputs with bounded maximum degree and derived positive results beating the greedy algorithm even in the lack of preemption.


\section{Technical Overview}

Throughout this paper, we study deterministic online preemptive fractional matching algorithms. In the fractional setting, the algorithm maintains a fractional matching by assigning weights $w : E \rightarrow [0,1]$ to edges, subject to the constraint that the total weight of edges incident to any vertex is at most~$1$. The objective is to maximize the sum of the weights of all edges. Upon the arrival of an edge, the algorithm may assign it an arbitrary weight satisfying this constraint. The algorithm is also allowed to preempt edge weights, i.e., to reduce the weight of any edge at any time. 

We say that a graph together with an edge-arrival instance is \emph{$\alpha$-hard} if $\alpha$ is the best competitive ratio achievable by any algorithm on this instance. Note that for bipartite graphs there is no integrality gap between the maximum size integral and fractional matchings. From a hardness perspective its also important to note that for any randomized integral preemptive algorithm there exists a deterministic fractional variant with the same competitive ratio.

For ease of exposition, we assume that in addition to adding vertices and edges, the adversary may also \emph{freeze} vertices. By freezing a vertex~$v$, the adversary guarantees that no additional edges incident to~$v$ will arrive in the future. Consequently, once a vertex is frozen, the algorithm has no incentive to reduce the weight of edges incident to~$v$. We refer the reader to Section~\ref{sec:frozen} for a formal discussion.

\subsection{Warm-up: Improving upon $2-\sqrt{2}$ Hardness}

\paragraph{Existing Constructions.}
Our approach builds on the following observation: existing hard instances for this problem do not merely establish an upper bound on the competitive ratio, but also force the algorithm to produce solutions with a very rigid structure. In particular, in both~\cite{EpsteinLSW18Weaker,HuangPTTWZ19Stronger}, the $4n$ vertices of the bipartite input graph can be partitioned into four equal-sized independent sets: $T_1,T_2$ (the \emph{tight} vertices) and $L_1,L_2$ (the \emph{loose} vertices). Edges exist only between the pairs $T_1$--$T_2$, $T_1$--$L_1$, and $T_2$--$L_2$ (see Figure~\ref{fig1}), and the graph admits a perfect matching between tight vertices $T_1 \cup T_2$ and loose vertices in $L_1 \cup L_2$.

Although not proven in their respective papers, it can be shown that both constructions share the following property: in order to achieve the optimal competitive ratio~$\alpha$, the algorithm must output a solution in which the tight vertices have average weight arbitrarily close to~$1$. Therefore, if our goal is to prove any strict separation from the previous best upper bound $\alpha = 2-\sqrt{2}$, we may safely assume that, when applied to these instances, the algorithm produces a solution in which the tight vertices have average weight close to~$1$.

\paragraph{Exploiting Structure.}
For simplicity, suppose we insert such an instance and the algorithms output is completely uniform: vertices in $T_1 \cup T_2$ have average weight~$1$, while vertices in $L_1 \cup L_2$ have average weight $2\alpha - 1$. Our goal is to penalize the algorithm for assigning weight~$1$ to the tight vertices. We freeze all vertices in $T_2$, and introduce a new set of vertices~$U$ of size $|T_1|$, connected to $T_1$ by a complete bipartite graph. Since vertices in $T_2$ are frozen, the algorithm doesn't want to preempt weight from edges incident to them. Consequently, the only way to assign weight to edges between $T_1$ and~$U$ is by preempting weight from edges between $T_1$ and $L_1$.

Regardless of the algorithm’s preemptive decisions, once $U$ and its incident edges are introduced, either the vertices in $L_1$ or those in~$U$ must have average weight at least $(2\alpha-1)/2$ (we can safely assume the algorithm keeps the vertices of $T_1$ tight). For simplicity, assume that the vertices in $L_1$ end up with uniform average weight $(2\alpha-1)/2$. We then freeze vertices of $T_1$ and insert a copy of the initial hard instance between the vertices of~$L_1$.

At this stage, the graph admits a perfect matching: vertices in $T_1$ can be matched to $U$, and vertices in $L_1$ can be matched among themselves. Thus, the maximum matching size is $5n/2$. Prior to the second insertion of the initial instance within $L_1$, the algorithm cannot increase its total weight by shifting weight from $T_1$--$L_1$ edges to $T_1$--$U$ edges. Hence, before the final step, the algorithm’s total weight is at most $2\alpha n$, since the initial instance is $\alpha$-hard.

When the $\alpha$-hard instance is recursively inserted into $L_1$, the vertices of $L_1$ already carry 'frozen' average weight $(2\alpha-1)/2 = r$. Intuitively, the algorithm can increase this weight by at most $\alpha(1-r)$, as the subgraph induced by $L_1$ is itself $\alpha$-hard. Therefore, the competitive ratio of the output at this point is upper bounded by
\[
\frac{2\alpha n + \frac{\alpha n}{2}\left(1 - \frac{2\alpha-1}{2}\right)}{5/2 n} < \alpha. 
\]

\begin{figure}[h]
    \centering
    \includegraphics[width=0.9\textwidth]{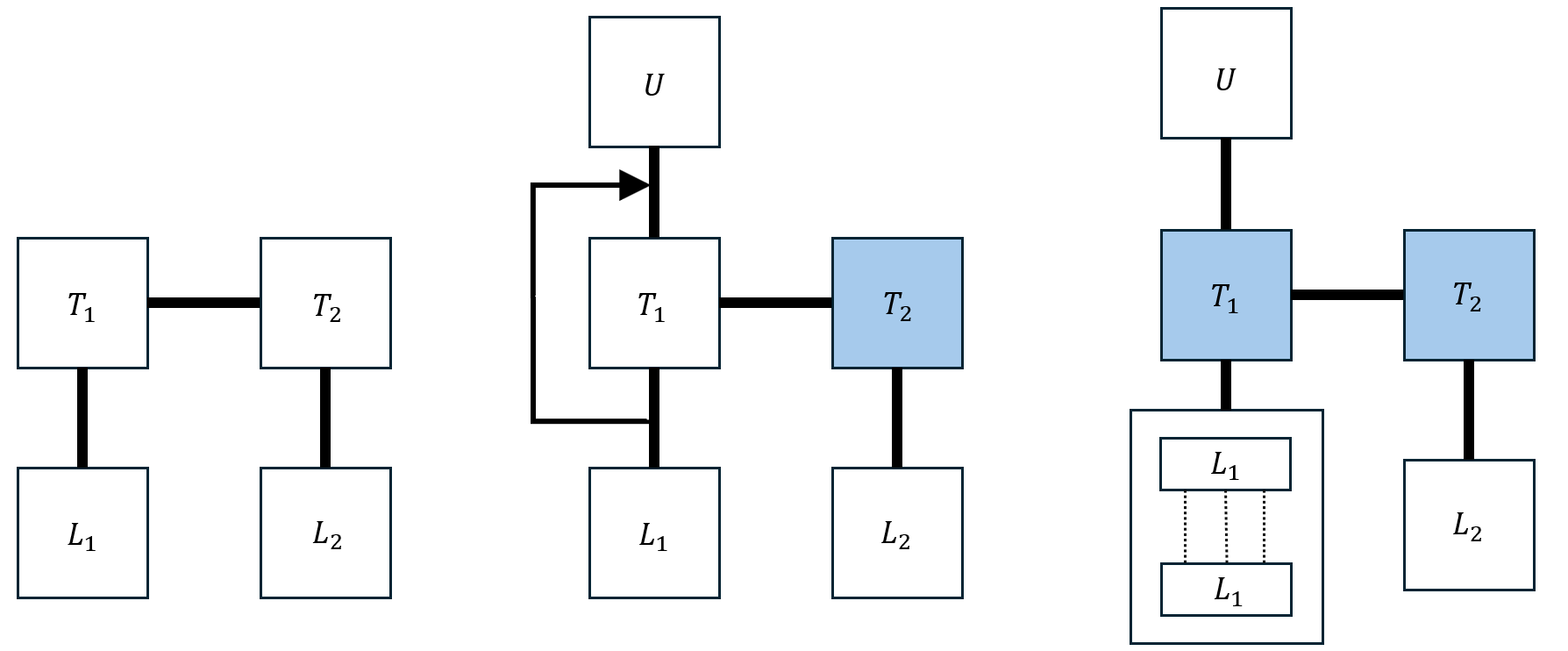}
    \caption{Left: Tight-loose structure Middle: Preemption after inserting U Right: Recursion in $L_1$}
    \label{fig1}
\end{figure}

\subsection{Showing a Strict Separation}

A primary goal of this paper is to demonstrate that online preemptive matching is strictly harder than online matching under vertex-future edge arrivals. Unfortunately, the best known competitive ratio for the latter, when restricted to integral algorithms, is approximately $0.567$, as shown by Huang et al.~\cite{HuangPTTWZ19Stronger}. Consequently, it is not sufficient to prove an arbitrary constant separation from the previous hardness bound of $2-\sqrt{2}$. To overcome this barrier, our construction employs several additional techniques, which we now outline at a high level.

\paragraph{Working with Non-Tight Tight Vertices.}
In the simplified construction above, the algorithm may slightly deviate from the assumed structure. We therefore establish a strong relationship between the extent to which tight vertices deviate from being fully tight and the resulting loss in competitive ratio. In Section~\ref{sec:bound}, we show that if the algorithm produces a solution in which tight vertices have weight at most $1-x$, then its competitive ratio degrades to $\alpha(1-x/2)$, for the $\alpha = 1/(1+\ln 2) \approx 0.59$ hard instance of Epstein et al.~\cite{EpsteinLSW18Weaker}. This requires reformulating their analysis in the fractional setting and generalizing it to algorithms that do not necessarily optimize their competitive ratio.

\paragraph{Scaling $|U|$.}
Further improvements come from scaling the size of $U$ to $\lambda |T_1|$ for an integer $\lambda > 1$. In this case, during the recursive step we insert the $\alpha$-hard instance among the $\lambda |T_1|$ heaviest vertices in $U \cup L_1$. Intuitively, increasing the size of $U$ causes a larger fraction of the initial weight of $L_1$ to reside on vertices that block the algorithm’s ability to improve its solution when the hard instance is inserted for the second time.

\paragraph{Optimizing the Recursive Argument.}
Observe that the final step of the construction does not rely on the internal structure of the underlying $\alpha$-hard instance. Instead, the argument implies the existence of a strictly harder instance with competitive ratio $\alpha' < \alpha$. This means the resulting hard instance can be substituted into the last step of the construction, allowing for a recursive formalization. We formalize this recursion scheme in Lemma~\ref{lem:recursion}, and in Section~\ref{sec:concludingmain} we approximate the strongest competitive ratio bound achievable by this approach.

\paragraph{Further Discussion.}
The upper bound of $0.5661$ established in this paper can likely be improved through a more refined analysis. Since our primary goal was to demonstrate a strict separation compared to the limitations of existing approaches using the simplest possible construction, we did not fully explore all possible approaches. For instance, this is why Lemma~\ref{lem:tightbaseinstance} relies on the simpler (but weaker) hard instance of~\cite{EpsteinLSW18Weaker}, rather than the more intricate state-of-the-art construction of~\cite{HuangPTTWZ19Stronger}.

We would like to highlight the fact that the previous upper bound of $2-\sqrt{2}$ by Huang et al. \cite{HuangPTTWZ19Stronger} is a constant that serves as a significant barrier in the approximation for the matching problem in multiple computational models. Namely, Konrad and Naidu \cite{DBLP:conf/approx/KonradN21} have shown a $2$-pass semi-streaming streaming algorithm and multiple papers \cite{DBLP:conf/soda/Behnezhad23, DBLP:journals/jacm/BhattacharyaKSW24, DBLP:conf/soda/AzarmehrBR24} have derived fully dynamic poly-logarithmic update time algorithms with the same approximation ratio.

It might also not be instantly clear why is it the case that if we are to insert an $\alpha$-hard instance between in independent set of size $n$ with frozen weight $r$ the algorithm may only improve its output by an additive factor of $\alpha(1-r)$ if the underlying vertices don't share the same frozen weight. We overcome this difficulty through creating multiple copies of the instance and treating the independent sets as meta vertices (for a formal derivation see Section~\ref{sec:meta}).

Finally, the reader may observe that the simplified construction above yields a non-bipartite input graph. In the full construction, instead of inserting the $\alpha$-hard sub-instance within $U$, we insert it between $U$ and a corresponding set $U'$ belonging to a copy of the initial instance (see Section~\ref{sec:recursion}), thereby preserving bipartiteness.

\section{Preliminaries}

Let $G = (V,E)$ be a graph. A \emph{fractional matching} in $G$ is a function
$x : E \to [0,1]$ such that for every vertex $v \in V$,
\[
\sum_{e \in \delta(v)} x_e \le 1,
\]
where $\delta(v)$ denotes the set of edges incident to $v$. We will refer to the weight of a vertex $v$ as $x_v = \sum_{e \in \delta(v)} x_e$. We will sometimes refer to the total weight of edges incident on a set of vertices $S \subseteq V$ by $x_S = \sum_{v \in S}x_v$.

\medskip

\begin{definition}[Online Preemptive Fractional Matching with General Edge Arrival]
Let $G=(V,E)$ be a graph whose edges arrive online in a sequence
$\sigma = (e_1, e_2, \ldots, e_T)$. An \emph{online preemptive fractional
matching algorithm} $\mathcal{A}$ maintains a fractional matching $x : E \to
[0,1]$ in an online fashion as follows:

\begin{enumerate}
    \item \textbf{Online arrival:} At each step $t$, the edge $e_t \in E$ arrives,
    revealing its endpoints. The algorithm
    must decide on a fractional assignment $x_{e_t} \in [0,1]$ for this edge.

    \item \textbf{Fractional matching constraint:} At all times, the algorithm
    maintains a fractional matching:
    \[
        \sum_{e \in \delta(v)} x_e \le 1 \quad \text{for all } v \in V,
    \]
    where $\delta(v)$ denotes the set of edges incident to vertex $v$ whose
    fractional weight has been assigned so far.

    \item \textbf{Preemption:} The algorithm is allowed to \emph{decrease} the
    fractional weight $x_e$ of any previously assigned edge $e$ at any time.
\end{enumerate}

The goal of the algorithm is to maximize the total fractional matching value
\[
\sum_{e \in E} x_e,
\]
subject to the above constraints.
\end{definition}
\begin{definition}
Let $\mathcal{A}$ be an online algorithm for an online matching problem, and
let $\sigma = (e_1, e_2, \ldots, e_T)$ be an input sequence of arriving edges.
For each $t = 1, \ldots, T$, let $\mathrm{ALG}(\sigma_t)$ and $\mathrm{OPT}(\sigma_t)$
denote the values of the matching produced by $\mathcal{A}$ and by an optimal
offline algorithm on $\sigma_t$, respectively.

The competitive ratio of $\mathcal{A}$ on the sequence $\sigma$ is defined as
\[
\mathrm{CR}_{\mathcal{A}}(\sigma)
\;=\;
\min_{1 \le t \le T}
\frac{\mathrm{ALG}(\sigma_t)}{\mathrm{OPT}(\sigma_t)}.
\]
\end{definition}

\section{Our Hard Construction}
\label{sec:bound}

In this section, we present our hard construction for the online preemptive fractional matching problem and its analysis. We begin by introducing several formal definitions that will be used throughout the remainder of the section.

\medskip

\begin{definition}
\label{def:tight-loose}
A graph $G = (V,E)$ on $4n$ vertices is said to have a \emph{tight-loose structure} if $V$ can be partitioned into four equal-sized, disjoint independent sets $T_1,T_2,L_1,L_2$ such that:
\begin{itemize}
    \item There exist perfect matchings in $G$ between $T_1$ and $L_1$, and between $T_2$ and $L_2$, and
    \item The edge set satisfies
    \[
        E = E[T_1 \times T_2] \;\cup\; E[T_1 \times L_1] \;\cup\; E[T_2 \times L_2].
    \]
\end{itemize}
\end{definition}

Much of our analysis focuses on specific edge arrival sequences on tight--loose graphs, under which any algorithm is forced to output a highly structured solution with weak competitive guarantees. This behavior is formalized by the following definition.

\medskip

\begin{definition}
\label{def:alphatight}
Let $G = (T_1 \cup T_2 \cup L_1 \cup L_2, E)$ be a tight-loose structured graph on $4n$ vertices together with an associated edge arrival sequence. We say that this pair forms an \emph{$(\alpha,f)$-tight-loose online preemptive fractional matching instance} for constant $\alpha \in [1/2,1]$ and function $f : [0,1] \rightarrow [0,1]$ if any online preemptive fractional matching algorithm observing the arrival sequence must output a fractional matching $w$ satisfying the following condition:
\begin{itemize}
    \item If the total weight assigned to vertices in $T_1 \cup T_2$ is
    \[
        w_{T_1 \cup T_2} = 2n(1-x)
    \]
    for some $x \in [0,1]$, then the total weight assigned to vertices in $L_1 \cup L_2$ satisfies
    \[
        w_{L_1 \cup L_2} \;\le\; 2n \cdot \bigl( 2\alpha(1-f(x)) -1 + x + o(1) \bigr).
    \]
\end{itemize}

Equivalently, if the algorithm assigns an average weight of $1-x$ to vertices in $T_1 \cup T_2$, then the resulting solution can be at most $\alpha(1-f(x)) + o(1)$ competitive. We assume throughout that the algorithm is aware of the bipartite decomposition of the input graph.
\end{definition}

We next formalize the ultimate objective of this section, namely the construction of edge arrival sequences under which no algorithm can achieve a strong competitive ratio.

\medskip

\begin{definition}
\label{def:alpha-hard}
A bipartite graph $G = (V,E)$ together with an associated edge arrival sequence is called an \emph{$\alpha$-hard online preemptive fractional matching instance} if every online preemptive fractional matching algorithm observing the sequence must output a solution whose competitive ratio is at most $\alpha + o(1)$ and $G$ admits a perfect matching. As before, we assume that the algorithm is aware of the bipartite decomposition of the input graph. 
\end{definition}

\paragraph{Outline of the Argument.}
We now state two lemmas which together imply our main result. The first lemma establishes the existence of a tight-loose instance with suitable parameters. Its proof appears in Section~\ref{sec:tightstructure}.

\medskip

\begin{lemma}
\label{lem:tightbaseinstance}
There exists a $\bigl(\tfrac{1}{1+\ln 2},\, x/2\bigr)$-tight--loose online preemptive fractional matching instance.
\end{lemma}

Our hard construction builds upon the tight-loose instance guaranteed by Lemma~\ref{lem:tightbaseinstance} through which we derive a recursive upper bound on the competitive ratios achievable by online preemptive fractional matching algorithms. In Section~\ref{sec:concludingmain}, we show how this recursion yields the bound stated in Theorem~\ref{thm:main}. The recursive argument itself is formalized and proved in Section~\ref{sec:recursion}.

\medskip

\begin{lemma}
\label{lem:recursion}
Assume that there exists a bipartite $\alpha$-hard online preemptive fractional matching instance for some $\alpha \le \tfrac{1}{1+\ln 2}$. Fix an arbitrary $\alpha_0 \in [1/2,\alpha]$, and define
\[
x_{\max} = 2\bigl(1 - \alpha_0(1+\ln 2)\bigr).
\]
Let
\[
\alpha' =
\min_{\lambda \in \mathbb{Z}_{>0}}
\max_{x \in [0,x_{\max}]}
\frac{
4 \cdot \frac{1-x/2}{1+\ln 2}
+ 2x
+ \lambda \alpha \left(1 - \frac{2 \cdot \frac{1-x/2}{1+\ln 2} - 1 + 2x}{\lambda + 1}\right)
}{
\lambda + 4
}
+ o(1).
\]
Then there exists a bipartite $\max(\alpha_0,\alpha')$-hard online preemptive fractional matching instance.
\end{lemma}

\paragraph{Arguing about Preemption}
\label{sec:frozen}

Throughout this section, for analytical convenience, we make several assumptions about the behavior of online preemptive fractional matching algorithms. These assumptions do not restrict generality and are introduced solely to simplify the competitive-ratio analysis.

Our first assumption is that the adversary may \emph{freeze} vertices. By freezing a vertex, the adversary guarantees that no further edges incident to that vertex will be revealed. Consequently, we may assume without loss of generality that the algorithm never preempts weight from edges incident to frozen vertices. Since frozen vertices cannot receive additional weight in the future, preempting edges incident to them cannot improve the algorithm’s competitive ratio. Formally, for any algorithm that preempts weight from frozen vertices, there exists an alternative algorithm that does not do so and achieves no worse competitive performance.

We further assume that whenever an edge is inserted, the algorithm assigns weight so as to make at least one of its endpoints tight. The algorithm has no incentive not to follow this strategy, as any weight assigned in this manner can be preempted later if needed. Moreover, after the insertion of an edge, we assume that the algorithm only preempts weight from edges incident to the tight endpoints of that edge (i.e., from vertices whose loads may increase as a result of the insertion), and not from elsewhere in the graph. This assumption is without loss of generality: preempting weight from unrelated edges cannot improve the algorithm’s competitive ratio at the current time, and any such preemption can always be deferred to a later point.

Finally, we assume that whenever the algorithm is faced with a choice between assigning total weight $x$ or $x+\Delta$ to a vertex set $U$, and this choice does not affect the weights of any other vertices, the algorithm always chooses the latter option. Although this assumption may appear odd and self-evident, some of our arguments rely on the algorithm placing excessive weight on certain vertices, and one might wonder whether the algorithm could avoid such situations by choosing the smaller assignment.

To justify this assumption, assume that in such cases the adversary may, for analytical purposes, introduce a collection of dummy vertices connected to $U$, allowing the algorithm to increase the total weight on $U$ to $x+\Delta$, while guaranteeing that the optimal offline solution does not use any of these edges. Since the algorithm may later preempt these edges at no cost, it has no rational reason not to exploit this additional flexibility. 

\subsection{Analyzing the Tight-Loose Structure of \cite{EpsteinLSW18Weaker}}
\label{sec:tightstructure}

In this section, we formally show that any online fractional algorithm observing the hard instance of Epstein et al.~\cite{EpsteinLSW18Weaker} must incur a loss in competitive ratio if it employs preemption in a way that deviates from the tight--loose structure. This establishes Lemma~\ref{lem:tightbaseinstance}. Our proof can be viewed as an extension of the original hardness argument of Epstein et al. to the setting of fractional algorithms where the algorithm does employ preemption.

\paragraph{The construction of Epstein et al. \cite{EpsteinLSW18Weaker}.}

The instance consists of $L$ layers $V_\ell : \ell \in [L]$ of vertices, where each layer contains $2n$ vertices. Through the arrivals the vertices of each layer $V_\ell$ will be partitioned into sets loose $L_\ell$ and ordered set tight $T_\ell$, with $|L_\ell| = |T_\ell| = n$ once the sequence finished. This partitioning and ordering will depend on the choices of the algorithm for all layers except layer $1$. For layer 1 half of its vertices fall in each of $T_1,L_1$ and $T_1$ is ordered arbitrarily. We will describe this partitioning dynamically as edges arrive. In the final input graph all edges will run between vertex sets $T_\ell$ and $V_{\ell +1}$ for $\ell \in [L-1]$. Note that there will be no edges between $L_\ell$ and $V_{\ell +1}$ for $\ell \in \{2 \dots L-1\}$.

Edges arrive in phases where in phase $\ell$ the edges between $T_\ell$ and $V_{\ell+1}$ arrive. At the start of phase $\ell$ the adversary sets $T_{\ell+1} = V_{\ell+1}$ and $L_{\ell + 1} = \emptyset$. In each update the adversary chooses the next vertex $v$ in $T_\ell$ according to the ordering and inserts edges from $v$ to all vertices of $T_{\ell +1}$. Then it moves the lowest weight vertex of $T_{\ell+1}$ to $L_{\ell+1}$. At the end of the phase the adversary decreasingly orders the remaining $n$ vertices of $T_{\ell+1}$ according to their current weight.

The final graph consists of the following:
\begin{itemize}
    \item Complete sub-graphs between vertices of $T_{\ell}$ and $T_{\ell + 1}$ for all $\ell \in [L-1]$
    \item An upper triangular graph between the vertices of $T_\ell$ and $L_{\ell+1}$ for $\ell \in [L-1]$
\end{itemize}

\paragraph{Dealing with Preemption.}

While reasoning about preemption is already challenging when an algorithm’s sole objective is to maximize its competitive ratio, our setting requires a more delicate analysis. In particular, we must establish a trade-off between the maximum competitive ratio an algorithm can achieve and the degree to which it allows the tight vertices in the construction of Epstein et al.~\cite{EpsteinLSW18Weaker} to deviate from being fully tight. Ideally, an algorithm would like to place weight on edges between loose vertices; however, by construction, no such edges exist.

In this setting, the algorithm’s best available strategy is to place weight on edges connecting tight vertices to loose ones. Whenever the edges incident to a vertex $v \in T_\ell$ and the vertices of $V_{\ell+1}$ are revealed, the algorithm has an opportunity to do so, as it knows that some of the neighbors of $v$ in layer $\ell+1$ will eventually be classified as loose. We may safely assume that the algorithm assigns as much weight as possible to these revealed edges: any weight temporarily assigned to edges between tight vertices can always be preempted later if needed.

Since at all times the set $T_{\ell+1}$ contains at least $n$ vertices, the algorithm can always make $v$ tight. Thus, without loss of generality, we assume that immediately after the edges from $v$ to $V_{\ell+1}$ are revealed, the algorithm makes $v$ tight.

Observe that during phase~$\ell$, edges arrive only between layers $\ell$ and $\ell+1$, and no edges between layers $\ell+1$ and $\ell+2$ are yet present. Therefore, at this point, the algorithm could only consider preempting weight from edges between layers $\ell-1$ and $\ell$, or from those between $\ell$ and $\ell+1$. For analytical convenience, we further wish to assume that during phase~$\ell$ the algorithm only preempts weight from edges between layers $\ell-1$ and $\ell$, and not from edges between layers $\ell$ and $\ell+1$.

To justify this assumption, we allow the algorithm a temporary relaxation: during phase~$\ell$, it may increase the total weight assigned to vertices in layer $\ell+1$ beyond~$1$. However, when phase~$\ell+1$ begins and the adversary reveals the edges incident to some vertex $v \in T_{\ell+1}$ and the vertices of $V_{\ell+2}$, the algorithm is required to preempt weight from edges incident to $v$ until its total incident weight is at most~$1$. We refer to this forced preemption step, when it occurs, as the \emph{pruning} of a vertex in layer~$\ell+1$.

Under this relaxation, the algorithm is never incentivized to remove weight from edges between layers $\ell$ and $\ell+1$ during phase~$\ell$. Moreover, since in phase~$\ell$ edges are introduced only incident to vertices in $T_\ell$, the algorithm may only benefit from preempting weight from edges between $T_{\ell-1}$ and $T_\ell$, as there are no edges between $L_{\ell-1}$ and $T_\ell$.

\paragraph{Competitive Analysis.} Call the \emph{residual weight} of a vertex at some point in time to be $1$ minus its current weight, that is how much more load it can take. Further define the following notation:

\begin{itemize}
    \item Let $\Res(T_\ell)$ stand for the total residual weight of vertices with weight at most $1$ in $T_\ell$ at the start of phase $\ell$ before any pruning occurs,
    \item Let $\Pre(T_\ell)$ stand for the total weight preempted from vertices of $T_\ell$ during phase $\ell+1$ (including pruning if it occurs).
\end{itemize}

\medskip

\begin{claim}
\label{cl:tight:eq1}

$n - \Res(T_{\ell+1}) \geq \ln 2\Res(T_\ell) + \Pre(T_{\ell-1})/2 - o(n)$ for all $\ell \in [L-1]$.

\end{claim}

As the proof of Claim~\ref{cl:tight:eq1} is rather numerically heavy and it would disrupt the flow of presentation we defer it to Appendix~\ref{app:tightnodes}. 

We now show how can it be used to conclude Lemma~\ref{lem:tightbaseinstance}. Consider how much can the algorithm increase the weight of its output when the edges of any vertex $v \in T_\ell$ to $V_{\ell+1}$ are revealed. If $v$-s residual weight is $r$ it can place $r$ weight on its recently arrived edges. In addition it may preempt weight from some of its edges and place that weight somewhere else. However, this does not increase the total weight of its output. Hence, the total weight of the algorithms output is $\sum_{\ell \in [L]}\Res(T_\ell)$, refer to this value $\Alg$.

Let $x$ stand for $\sum_{\ell \in [L]}\Pre(T_{\ell})/(n \cdot L)$. After the algorithm preempts edge weights in phase $\ell$ from edges between layers $\ell-1$ and $\ell$ the $\ell-1$ can never recover that weight as it will not observe further edges. Hence, $x$ will stand for the average residual weight of tight $(\cup T_\ell)$ vertices at the end of the insertion sequence. By summing up the left hand side of Claim~\ref{cl:tight:eq1} over all layers we get:

\begin{align*}
 \sum_{\ell \in [L-1]} n - \Res(T_{\ell+1}) & =  nL + \Alg - \Res(T_1) \nonumber \\
& \leq  (n+1)L + \Alg \nonumber\\
\end{align*}

Here we used that $\Res(T_\ell) \leq n$ for any layer $\ell$. By summing up the right hand side we get:

\begin{align*}
\sum_{\ell \in [L-1]} \ln 2\Res(T_\ell) + \Pre(T_{\ell-1})/2 - o(n) & =  \ln 2 \Alg + nLx/2 - \ln 2 \Res(T_L) - \Pre(T_L) - o(nL) \nonumber \\
& \geq \ln2 \Alg + nLx/2 - 3n - o(nL) \nonumber
\end{align*}

Here we used that $\Pre(T_\ell) \leq n$ for any layer $\ell$ as the algorithm may only remove at most $n$ weight of the $n$ vertices of layer $\ell$. By evoking the claim and reordering we get that $\Alg \cdot (1+\ln(2)) \leq (n+1)L +3n - nLx/2 - o(nL)$. The graph contains a matching of size $(L-1)n$ Hence, the algorithm has competitive ratio 

$$\frac{(n+1)L + 2n -nLx/2 - o(nL)}{n(L-1) \cdot (1+\ln(2))} = \frac{1-x/2}{1+\ln 2} + o(1)$$

as we can select $L$ to be an arbitrarily large value.

\subsection{Proving the Recursive Formula}

This section is devoted to proving our recursive formula from Lemma~\ref{lem:recursion}. Recall that the lemma assumes there exists some $\alpha$-hard bipartite instance for online preemptive fractional matching. We further know from Lemma~\ref{lem:tightbaseinstance} that there exists a $(1/(1+\ln 2),x/2)$-tight-loose instance, which will be the starting point of the adversary.

After inserting the tight-loose structure let the resulting graph and its tight and loose vertex sets be denoted by $G = (T_1\cup T_2 \cup L_1 \cup L_2,E)$ where all four vertex sets have the same size $n$. Recall that by the definition of a tight-loose structure edges may only run between $T_1-T_2$, $T_1-L_1$ and $T_2-L_2$. Assume that the algorithm assigns an average weight of $1-x$ to vertices of $T_1 \cup T_2$.

At this point there is a perfect matching present in the instance and by the lemma we know that the algorithms competitive ratio can be at most $(1-x/2)/(1+\ln 2) + o(1)$, denote this value by $y$. This implies an upper bound on the average weight of loose vertices of $L_1 \cup L_2$ of $2y -1 +x$. It could be the case that the average weight of $T_1$ and $T_2$ vertices are not the same. Let these average weight values for the two sets be denoted by $1-x_1$ and $1-x_2$ respectively, where $x_1 + x_2 = 2x$.

\medskip

\begin{claim}
\label{cl:loosenodes}

The average weight of vertices in $L_1$ and $L_2$ are $2y-1+x_2$ and $2y-1+x_1$ respectively.

\end{claim}

\begin{proof}

Let $w_{T_1,T_2}$, $w_{T_1,L_1}$ and $w_{T_2,L_2}$ stand for the total weight of edges running between vertex sets $T_1,T_2$, $T_1,L_1$ and $T_2,L_2$ respectively. As all edges run between these bipartition, we have that $w_{T_1,T_2}+w_{T_1,L_1}+w_{T_2,L_2} = 2ny$. From the definition of $x_1$ and $x_2$ we also know that $w_{T_1,T_2}+w_{T_1,L_1} = n(1-x_1)$ and $w_{T_1,T_2}+w_{T_2,L_2} = n(1-x_2)$.

Combining the equalities we get that $w_{T_1,L_1} = n(2y-1+x_2)$ and $w_{T_2,L_2} = n(2y-1+x_1)$. As all edges of loose vertices run to their corresponding tight sets we can conclude the claim.

\end{proof}

At this point, the adversary creates a copy of the instance together with the algorithm’s current output. Whether such copying is permissible varies across the literature. In Section~\ref{sec:copy}, we explain why, from a hardness perspective, it is without loss of generality to assume that the adversary may duplicate instances—including the algorithm’s output—in the context of the online preemptive fractional matching problem.

Let the copied instance be denoted by $G(T_1' \cup T_2' \cup L_1' \cup L_2', E')$. The adversary then inserts vertex sets $U$ and $U'$ of size $\lambda n$ for some integer constant $\lambda$. At this point the algorithm freezes the vertices of $T_2,L_2, T_1',L_1'$. Afterwards the adversary inserts complete bipartite sub-graphs between vertices of $T_1-U$ and $T_2'-U'$. The algorithm will react to this insertions. After this reaction let the vertex sets $U_{L_1}$ and $U_{L_2'}'$ refer to the $\lambda n$ highest weight vertices of $L_1 \cup U$ and $L_2' \cup U'$ respectively.

\medskip

\begin{claim}
\label{cl:Uweight}

The average weights of vertices in both $U_{L_1}$ and $U_{L_2'}'$ are at least $(2y-1+2x)/(\lambda+1)$.

\end{claim}

\begin{proof}

Without loss of generality focus on the average weight of $U_{L_1}$. Consider how the algorithm may have reacted to the insertion of the complete bipartite sub-graph between $U$ and $T_1$. As $T_2$ is frozen at this point, the algorithm doesn't want to preempt weight from $T_2-T_1$ edges. The algorithm also wants to make at least one endpoint of all the edges in this complete bipartite graph tight to maximize its competitive ratio. The average weight of $U_{L_1}$ is minimized if the algorithm decides to make the $T_1$ endpoints tight.

Before the insertion the average weight of $T_1$ vertices is $1-x_1$ so the algorithm can safely place a total weight of $nx_1$ on the newly arrived edges to $U$. Afterwards, it might choose to preempt some weight from $T_1-L_1$ edges to distribute them into $T_1-U$ edges. Either way, the total weight of edges between $T_1$ and $L_1 \cup U$ will be $n(2y-1+x_1 + x_2) = n(2y-1+2x)$ by Claim~\ref{cl:loosenodes}. As $U_{L_1}$ consists of the $\lambda n$ highest weight vertices of $U \cup L_1$ they will be hosting at least a $\lambda/(\lambda+1)$ fraction of this weight implying the claim.

\end{proof}

Now the adversary freezes vertices $T_1, T_2'$ and creates further $n-1$ copies of the whole instance so far (lets refer to this as the \emph{medium stage}), in each including the output of the algorithm. Let $(U_{L_1})^1,\dots,(U_{L_1})^n$ and $(U_{L_2'}')^1, \dots ,(U_{L_2'}')^n$ correspond to the copies of $U_{L_1}$ and $U_{L_2'}'$ in these instances. Finally, the adversary inserts the initial $\alpha$-hard instance between the vertices of $\bigcup (U_{L_1})^i$ and $\bigcup (U_{L_2'}')^i$.

\begin{figure}[h]
    \centering
    \includegraphics[width=0.9\textwidth]{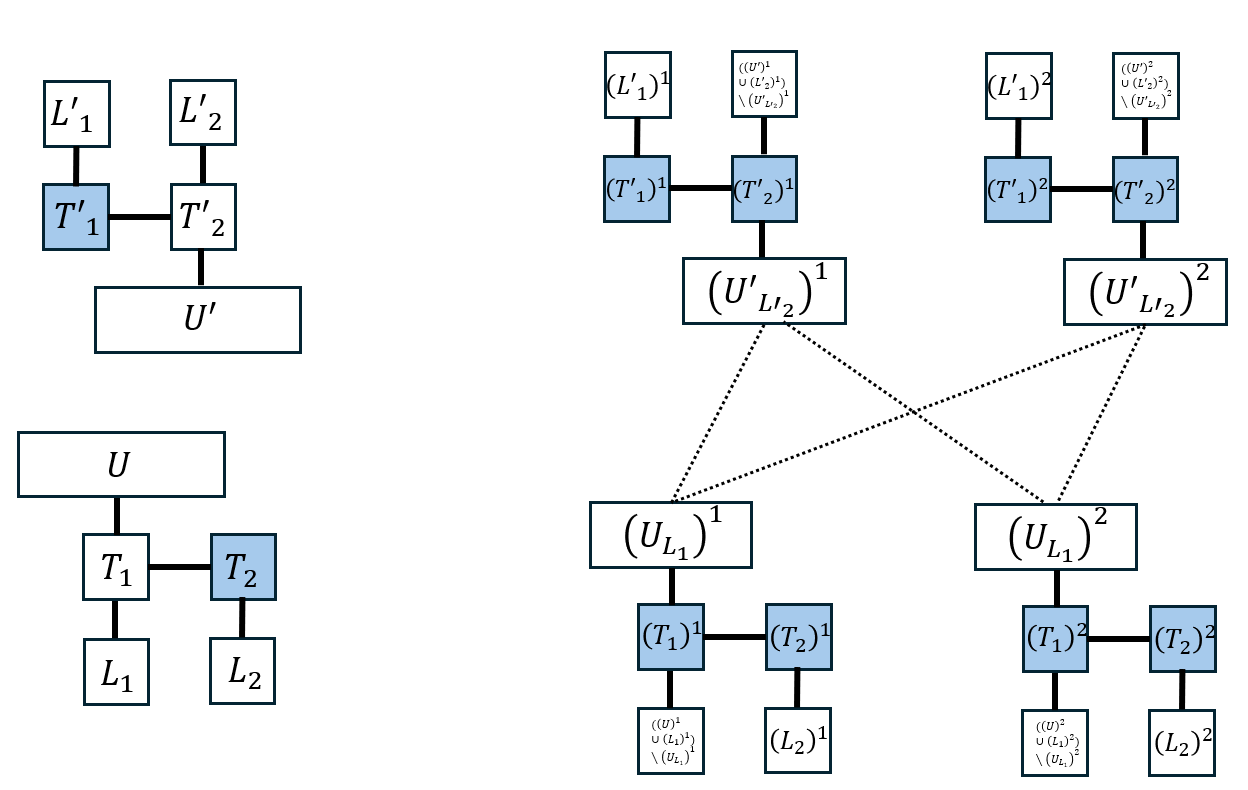}
    \caption{Left: Medium stage, Right: Final construction}
    \label{fig2}
\end{figure}

The detail of this process are formally explained in Section~\ref{sec:meta} and the guarantees of its outcome are summarized by Lemma~\ref{lem:meta}. Essentially, the algorithm treats vertex sets $(U_{L_1})^i,(U_{L_2'}')^i$ as meta vertices in a bipartite meta graph $\mathcal{G}$ and inserts the $\alpha$-hard instance into $\mathcal{G}$. The insertion of an edge between a pair of meta vertices is simulated through an insertion of a complete bipartite graph between the corresponding vertex sets. Now we turn to analyzing the competitive ratio.

\medskip

\begin{claim}

At the end of the edge arrival sequence the weight of an optimal offline solution is $n^2(4+\lambda)$.

\end{claim}

\begin{proof}

The resulting graph has $2(4+\lambda)n^2$ vertices. The initial instance was on $4n$ vertices which were first copied once. The addition of vertex sets $U,U'$ increased the vertex site size to a further $2(4+\lambda)n$ and the medium stage instance was copied $n-1$ times. So its sufficient to argue that there exist a perfect matching in the graph at the end. Note that in each copy of the medium stage instance vertex sets $T_2$ and $T_1'$ can be perfectly matched to vertices of $L_2$ and $L_1'$ respectively due to the definition of the initial tight-loose structure. In the meta graph from the definition of an $\alpha$-hard instance there must exist a perfect matching of the meta vertices. As meta edges represent a complete bipartite sub-graph between the corresponding vertex sets vertices of $\cup (U_{L_1})^i$ can be perfectly matched to vertices of $\cup (U'_{L_2'})^i$.

It remains to argue that for all copies of the medium stage instance vertex sets copies $T_1$ and $T_2'$ can be matched to their corresponding $(U \cup L_1) \setminus U_{L_1}$ and $(U' \cup L_2') \setminus U_{L_2'}'$ sets unmatched so far. Without loss of generality focus on the former pair. We know that there exist a perfect matching $M^*_{T_1-L_1}$ between $T_1$ and $L_1$ from the definition of the tight-loose structure and there is a complete bipartite sub-graph between $T_1$ and $U$. Let $M^*_{T_1-(L_1 \setminus U_{L_1})}$ be denoted to the restriction of this perfect matching to vertices of $L_1 \setminus U_{L_1}$. Match the $T_1$ endpoints of $M^*_{T_1-(L_1 \setminus U_{L_1})}$ through the edges of $M^*_{T_1-(L_1 \setminus U_{L_1})}$. The remaining vertices of $T_1$ can be arbitrarily matched within the complete subgraph $T_1 \cup (U \setminus U_{L_1})$.

\end{proof}

Finally, with the following claim upper bounding the performance of the algorithm, we may conclude the lemma.

\medskip

\begin{claim}
The algorithms output at the end of the edge arrival sequence can have total weight at most $n^2 \cdot (4y+2x +\lambda \alpha(1-\frac{2y-1+2x}{\lambda+1}) )$.

\end{claim}

\begin{proof}

After the initialization of the tight-loose instance by Lemma~\ref{lem:tightbaseinstance} the competitive ratio of the algorithm is at most $y$, hence the algorithms output has weight at most $2yn$. Once its copied the first time this jumps to $4yn$. When the complete bipartite sub-graphs between $U,T_1$ and $U',T_2'$ are inserted the algorithm may increase its output weight by $n(x_1+x_2) = 2nx$ to $4yn + 2nx$. At this point the medium stage instance is copied $n-1$ times blowing this value up by $n$.

Finally, by Lemma~\ref{lem:meta} when the $\alpha$-hard instance is inserted into the meta graph $\mathcal{G'}$ the algorithm may increase the total weight of its output by $n^2\alpha \lambda(1-\frac{2y-1+2x}{\lambda+1})$. Recall that all sets of vertices in $(U_{L_1})^i, (U_{L_2'}')^i : i \in [n]$ have average weight at least $\frac{2y-1+2x}{\lambda+1}$ before the insertion of the $\alpha$-hard instance by Claim~\ref{cl:Uweight}.

\end{proof}

\label{sec:recursion}

\subsection{Concluding Theorem~\ref{thm:main}}

\label{sec:concludingmain}

In this section we finally conclude our main Theorem~\ref{thm:main}. This is achieved through a precise approximation of what competitive ratio upper bound can be established using Lemma~\ref{lem:recursion}. Define $\alpha_0 = 0.5661$, the constant in our main theorem.

\medskip

   \begin{claim}\label{numeric}
        Let $x_{max} = 2 \cdot (1-\alpha_0 \cdot (1+\ln 2))$, then for any $\alpha\in[\alpha_0,\frac{1}{1+\ln(2)}]$ we have
        $$\min_{\lambda \in \mathbb{Z}_{>0}} \max_{x \in [0,x_{\max}]} \frac{4 \cdot \frac{1-x/2}{1 + \ln 2} +  2x + \lambda \alpha \cdot (1-\frac{2 \cdot \frac{1-x/2}{1 + \ln 2} -1 +2x}{\lambda + 1}) }{\lambda + 4}\leq \alpha - \epsilon,$$
        for fixed constant $\epsilon = 10^5$.
    \end{claim}

Observe that Claim~\ref{numeric} combined with Lemma~\ref{lem:recursion} and the $1/(1+\ln 2)$-hard instance of Epstein et al. \cite{EpsteinLSW18Weaker} implies our main theorem. If we set any $\alpha \in [\alpha_0, 1/(1+\ln2)]$ to be the competitive ratio shown by an $\alpha$-hard instance, the claim and lemma together imply the existence of an $(\alpha - 10^{-5})$-hard instance.

    \begin{proof}
       Define $y = \frac{1-x/2}{1 + \ln 2}$. Note that $\alpha_0 = \frac{1-\frac{x_{max}}{2}}{1 + \ln 2} $. The expression can now be re-stated as 
       $$\min_{\lambda \in \mathbb{Z}_{>0}} \max_{y \in [\alpha_0,\frac{1}{1+\ln(2)}]} \frac{4(1-y\ln(2)) + \lambda \alpha  (\frac{\lambda -2+y(2+4\ln(2))}{\lambda + 1}) }{\lambda + 4}-\alpha\leq  - \epsilon$$
Fix $\lambda = 129$. Note that the choice of the constant $\lambda$ is sensitive, but any sufficiently large but comparably smaller value then $1/\epsilon$ suffices. Consider the function,

$$f(\alpha,y)=\frac{4(1-y\ln(2)) + \lambda \alpha  (\frac{\lambda -2+y(2+4\ln(2))}{\lambda + 1}) }{\lambda + 4}-\alpha$$

In order to describe the tightest case of the initial min-max expression we need to consider when $f(\alpha,y)$ is maximized. Hence, we examine $\frac{d f(\alpha,y)}{d \alpha}$.

\begin{align*}
\frac{d f(\alpha,y)}{d \alpha} & = \frac{\lambda \cdot \left(\frac{\lambda-2+y(2+4\ln2)}{\lambda + 1}\right)}{\lambda + 4} - 1 \nonumber \\
& \leq \frac{6y-6}{\lambda + 4} \nonumber \\
& \leq \frac{\frac{6}{1+\ln 2} - 6}{\lambda + 4} \nonumber \\
& \leq 0 \nonumber
\end{align*}

The second inequality follows from $y \leq 1/(1+\ln 2)$. This implies that the expression is maximized when $\alpha$ is minimized so we can proceed with assuming $\alpha = \alpha_0$. $f(\alpha_0,y)$ is linear with respect to the parameter $y$. Consequently, it must obtain its maximum on one of the endpoints of the interval range so its sufficient to prove that

$$\max_{y \in \{\alpha_0,\frac{1}{1+\ln 2} \}} f(\alpha_0,y) \leq -10^{-5}$$

Via machine its easy to check that $f(\alpha_0,\alpha_0) \sim -1.12 \cdot 10^{-5}$ and $f(\alpha_0,1/(1+\ln 2)) \sim -2.81 \cdot 10^{-5}$ concluding the proof.  
    \end{proof}

\section{Technical Notes}
\label{sec:notes}

\subsection{Recursing on Meta-Graphs}
\label{sec:meta}

The final step in the proof of Lemma~\ref{lem:recursion} crucially relies on the assumption that an $\alpha$-hard instance can be used to bound how much additional weight an algorithm may place on vertices that already carry frozen weight. In this section, we formalize this assumption.

\medskip

\begin{lemma}
\label{lem:meta}
Suppose an online preemptive fractional matching algorithm is fed a bipartite graph $G = (L \cup R, E)$ together with an edge arrival sequence, and at some point produces a fractional matching $w : E \rightarrow [0,1]$. Assume further that the following conditions hold:
\begin{itemize}
    \item There exist pairwise disjoint independent sets $L_1,\dots,L_N \subseteq L$ and $R_1,\dots,R_N \subseteq R$, each of size~$k$;
    \item For every $i \in [N]$, all neighbors of vertices in $L_i \cup R_i$ have been frozen;
    \item For every $i \in [N]$, the average vertex weight induced by $w$ on $L_i \cup R_i$ is at least~$r$;
    \item There exists a bipartite $\alpha$-hard online preemptive fractional matching instance on $N$ vertices.
\end{itemize}
Then there exists a graph $(\bigcup_{i \in [N]} L_i \cup R_i, E_\alpha)$ and an associated edge arrival sequence such that:
\begin{enumerate}
    \item $E_\alpha$ contains a perfect matching between $\bigcup_{i \in [N]} L_i$ and $\bigcup_{i \in [N]} R_i$, and
    \item If the algorithm’s input is extended by this arrival sequence, the total weight of its output can increase by at most $(1-r)k\alpha N$.
\end{enumerate}
\end{lemma}

\begin{proof}
We construct a bipartite \emph{meta-graph} $\mathcal{G}$ with vertex sets $\mathcal{L}$ and $\mathcal{R}$, each of size~$N$. Each meta-vertex in $\mathcal{L}$ (respectively, $\mathcal{R}$) corresponds to one independent set $L_i$ (respectively, $R_i$) in~$G$.

We now insert the $\alpha$-hard instance into $\mathcal{G}$. Whenever an edge between meta-vertices corresponding to $L_i \in \mathcal{L}$ and $R_j \in \mathcal{R}$ arrives in the $\alpha$-hard instance, we insert a complete bipartite subgraph between the corresponding independent sets $L_i$ and $R_j$ in~$G$. This process defines the edge set $E_\alpha$ and its arrival sequence. Since the $\alpha$-hard instance admits a perfect matching, it follows that $E_\alpha$ contains a perfect matching between $\bigcup_{i \in [N]} L_i$ and $\bigcup_{i \in [N]} R_i$.

Next, we relate the behavior of the algorithm on $G$ to an induced fractional matching on $\mathcal{G}$. Whenever the algorithm increases or decreases the weight of an edge $(u,v)$ with $u \in L_i$ and $v \in R_j$ by an amount $\Delta \in [0,1]$, we apply a corresponding change of $\Delta / (k(1-r))$ to the weight of the edge between the meta-vertices representing $L_i$ and $R_j$, in the same direction.

We first observe that these updates define a valid execution of an online preemptive fractional matching algorithm on the meta-graph $\mathcal{G}$ under the $\alpha$-hard arrival sequence. Since all neighbors of vertices in each $L_i \cup R_i$ are frozen in~$G$, we may assume that the algorithm never preempts weight from edges outside of $E_\alpha$ during the arrivals of $E_\alpha$. Consequently, the average weight of vertices in each $L_i \cup R_i$ remains within the interval $[r,1]$ throughout this process, which implies that the induced vertex weights in $\mathcal{G}$ always lie in $[0,1]$.

Because the instance inserted into $\mathcal{G}$ is $\alpha$-hard, the total weight of the induced fractional matching on $\mathcal{G}$ is at most $\alpha N$. By construction, this implies that the total increase in weight across all edges of $G$ during the arrivals of $E_\alpha$ is at most $(1-r)k\alpha N$, completing the proof.
\end{proof}

\subsection{Comparing Vertex-Future Edge and Generalized Vertex Arrivals}
\label{sec:ranking}

In this section, we justify claims regarding the relationship between the vertex-future edge arrival model and the generalized vertex arrival model~\cite{WangW15}, as well as the competitive ratios achievable in the former.

In the standard definition of generalized vertex arrivals, vertices are associated with arrival times and deadlines. When a vertex arrives, it reveals all edges incident to previously arrived vertices. At its deadline, the algorithm must irrevocably decide whether to match the vertex to an available neighbor. It is typically assumed that by the time of its deadline, a vertex has observed all of its neighbors.

A generalized vertex arrival sequence can be simulated within the vertex-future edge arrival model by setting the arrival time of each vertex in the latter equal to its deadline in the former. Under this transformation, all edges incident to a vertex are revealed before the time at which a matching decision must be made, establishing the claimed relationship between the two models.

The converse, however, does not hold. In the generalized vertex arrival model, for a vertex to observe all of its neighbors by its deadline, all of those neighbors must have arrived beforehand and revealed their incident edges. In particular, even in bipartite graphs—where edges do not run between neighbors of the same vertex—this constraint persists: once all neighbors of a vertex have arrived, any subsequently arriving vertices must reveal all edges incident to the existing vertices. This restriction does not apply in the vertex-future edge arrival models.

Finally, we argue that known positive results for the generalized vertex arrival model extend to the vertex-future edge arrival model. The randomized integral and deterministic fractional algorithms of Huang et al.~\cite{HuangPTTWZ19Stronger} are based on the \emph{Ranking} and \emph{Water-Filling} frameworks. At a high level, these algorithms rely only on local information, namely the edges revealed so far and the matched status or weight of neighboring vertices. Since this information is also available in the vertex-future edge arrival model, their competitive ratio guarantees carry over unchanged.

\appendix

\section{Missing Proofs}

\subsection{Proof of Claim~\ref{cl:tight:eq1}}
\label{app:tightnodes}
During phase $\ell$ the algorithm can only place weight on edges between $T_\ell$ and $V_{\ell+1}$. Let $\Res(T_\ell,i)$ stand for the amount of weight the algorithm places on the edges to $V_\ell$ of the $i$-th vertex in the ordering of $T_\ell$ when exploiting its residual weight. Therefore, $\sum_{i \in [n]}\Res(T_\ell,i) = \Res(T_\ell)$. Further define, with a slight trick in notation, that $\Pre(T_\ell,i)$ stands for amount the weight the algorithm places on the same edges through preempting from the previous layer $\ell-1$. Hence, $\sum_{i \in [n]} \Pre(T_\ell,i) = \Pre(T_{\ell-1})$.

At the start of the phase vertices of $V_{\ell + 1}$ have weight $0$ and they are all assigned to $T_{\ell+1}$. As edges of vertices of $T_\ell$ arrive the lightest vertex of $T_{\ell+1}$ always moves to $L_{\ell + 1}$. Let $w(i)$ stand for the weight of the $i$-th vertex of $V_{\ell+1}$ which was moved to $L_{\ell+1}$.

Up until and including the time when the algorithm has reached the $k$-th vertex in $T_L$ the algorithm has placed $\sum_{i=1}^{k}\Pre(T_\ell,i) + \Res(T_\ell,i)$ weight on the edges running between $T_\ell$ and $V_{\ell+1}$. Out of $V_{\ell +1}$ the set $L_{\ell+1}$ contains the $k-1$ lightest vertices when the $k$-th vertex is added to $L_{\ell+1}$. As at this point there are $2n-i+1$ vertices in $T_{\ell+1}$:

\begin{align}
w(k) \leq \frac{(\sum_{i=1}^{k}\Pre(T_\ell,i) + \Res(T_\ell,i) - w(i)) + w(k)}{2n-i + 1}\quad\text{for all}\quad k \in [n] \label{eq:tight:1} 
\end{align}

For $k\in[n]$, define
\[
W_k := \sum_{i=1}^k w(i)
\qquad\text{and}\qquad
A_k := \sum_{i=1}^k \bigl(\Pre(T_\ell,i)+\Res(T_\ell,i)\bigr).
\]
\[
S_k := A_k - W_k = \sum_{i=1}^k \bigl(\Pre(T_\ell,i)+\Res(T_\ell,i)-w(i)\bigr).
\]

Inequality~\ref{eq:tight:1} could be re-stated as,
\begin{align}
w(k) &\le \frac{S_k - w(k)}{2n-k+1} \nonumber\\
&=   \frac{A_k - W_{k-1}}{2n-k+1}\label{eq:tight:2}
\end{align}

Our goal will be to upper bound $W(n)$. By Inequality~\ref{eq:tight:2}

\begin{align}
W(n) & = W(n-1)+w(n) \nonumber \\
& \leq \frac{2n-n}{2n-n+1}W(n-1)
+ \frac{A_k}{2n-n+1} \nonumber \\
& \leq \sum_{i=1}^n
\frac{A_i}{2n-i+1}
\prod_{j=i+1}^n \frac{2n-j}{2n-j+1} \nonumber \\
& \leq \sum_{i=1}^n
\frac{n-i-1}{2n-i+1}\,A_i \nonumber
\end{align}

Here we exploited that the product telescopes. We define $c_i := \frac{n-i-1}{2n-i+1}$ and open up $A_i$ to get,

\begin{equation}
W(n)
\le
\sum_{t=1}^n c_i\bigl(\Pre(T_\ell,i)+\Res(T_\ell,i)\bigr) \nonumber
\end{equation}

Since $c_i \le \tfrac12$ for all $i$ for the term concerning preemption we get,
\begin{equation}
\sum_{i=1}^n c_i\,\Pre(T_\ell,i)
\le
\frac12\sum_{i=1}^n \Pre(T_\ell,i)
= \frac{1}{2}\Pre(T_\ell)
\nonumber
\end{equation}

In order to upper bound the term coming from residual weights observe that the sequence $c_i$ is decreasing, while $\Res(T_\ell,i)$ is increasing. By Abel summation,
\[
\sum_{i=1}^n c_i\,\Res(T_\ell,i)
=
c_n\sum_{i=1}^n \Res(T_\ell,i)
+
\sum_{t=1}^{n-1}(c_t-c_{t+1})
\sum_{i=1}^t \Res(T_\ell,i).
\]
Since $\Res(T_\ell,i)$ is increasing and bounded by $1$,

\begin{equation}
\sum_{i=1}^n c_i\,\Res(T_\ell,i)
\le
\frac1n\Bigl(\sum_{i=1}^n c_i\Bigr)
\sum_{i=1}^n \Res(T_\ell,i).
\nonumber
\end{equation}

By properties of harmonic series we have that,
\[
\sum_{i=1}^n c_i
=
n-(n+2)\sum_{m=n+3}^{2n+2}\frac1m
=
n(1-\ln 2)+O(1).
\]
Substituting this into the previous equation we get:
\begin{equation}
\sum_{i=1}^n c_i\,\Res(T_\ell,i)
\le
(1-\ln 2)\sum_{i=1}^n \Res(T_\ell,i)
+o(n).
\nonumber
\end{equation}

Combining our bounds on the preemptive and residual values we can finally upper bound $W(n)$

\begin{align}
W(n) & = \sum_{i=1}^n w(i) \nonumber \\
& \le (1-\ln 2)\sum_{i=1}^n \Res(T_\ell,i) + \frac12\sum_{i=1}^n \Pre(T_\ell,i) + o(n) \nonumber \\
& = (1-\ln 2) \Res(T_\ell) + \frac12\Pre(T_{\ell-1}) \nonumber
\end{align}

During phase $\ell$ the algorithm distributes of $\Res(T_\ell) + \Pre(T_{\ell-1})$ on the edges running between layer $\ell$ and $\ell+1$, and during phase $\ell$ it does not preempt from the same edges. As at most $W(n)$ of this weight is placed on vertices of $L_{\ell+1}$ and the rest is on $T_{\ell+1}$ we can conclude that,

\begin{align}
n - \Res_T(\ell+1) & \geq \Res(T_\ell) + \Pre(T_{\ell-1}) - W(n) \nonumber \\
& \geq \ln 2\Res(T_\ell) + \frac12\Pre(T_{\ell+1}) - o(n)\nonumber
\end{align}

\subsection{Copying the Instance}
\label{sec:copy}

Certain steps in the proof of Lemma~\ref{lem:recursion} rely crucially on the assumption that, at any point in time, the adversary may create an arbitrary number of copies of the current input instance on which the algorithm behaves identically. As this assumption may appear nonstandard, we justify it here at a high level.

We note that our recursive argument can, in principle, be carried out even without assuming that the resulting copies are perfectly identical. We nevertheless adopt this assumption in order to simplify the presentation.

Suppose the adversary wishes to create $N$ copies of the current state of the input graph and the algorithm’s output on an instance with $n$ vertices, in order to establish an upper bound of~$\alpha$ on the competitive ratio of a deterministic online preemptive fractional matching algorithm. For simplicity, assume that the value of an optimal offline solution for the instance is at least $\Omega(n)$.

We first consider a discretized variant of the online preemptive fractional matching problem. In this model, edges arrive online as usual, and the algorithm may update its solution both online and via preemption. However, every update must assign to each edge a weight that is a multiple of $1/\gamma$, for some sufficiently large predetermined parameter~$\gamma$.

Under this discretization, after the arrival of an edge the algorithm has only finitely many possible actions: it may increase or decrease the weight of any of the $O(n^2)$ edges by an integer multiple of $1/\gamma$. As a result, after any update the algorithm has at most $O(\gamma^{O(n^2)})$ possible distinct actions. Since there are $O(n^2)$ edge arrivals in total, the algorithm can produce at most
\[
O\bigl(n^{\gamma^{O(n^2)}}\bigr)
\]
distinct outputs over the entire execution.

The precise bound is not important; what matters is that it is a function of $n$ and~$\gamma$. We now choose $\gamma = \mathrm{poly}(n)$ with $\gamma \gg n^2$. Observe that for any continuous fractional algorithm, there exists a discretized algorithm that at all times produces a fractional matching differing by at most $1/\gamma$ per edge.

If the original (continuous) algorithm has competitive ratio~$\alpha$, then the discretized algorithm achieves competitive ratio $\alpha - O(n^2/\gamma) = \alpha - o(1)$. Thus, from a hardness perspective, we may assume without loss of generality that the algorithm is discretized using a sufficiently large parameter~$\gamma$. In particular, for a fixed input sequence on $n$ vertices, the algorithm can reach only some large but finite number $f(n)$ of distinct outputs.

Now suppose the adversary wishes to force the algorithm to produce $N$ identical copies of its output. It suffices for the adversary to generate $f(n)^2 N^3$ independent copies of the current input instance. We may further assume that after each edge insertion, the algorithm only updates weights on incident edges, as any additional preemptions can be postponed to a later time. Under this assumption, across all $f(n)^2 N^3$ copies, the discretized algorithm can behave in at most $f(n)$ distinct ways.

We classify copies into \emph{types}, where two copies belong to the same type if the algorithm produces identical fractional matchings on them. Any type containing fewer than $f(n)N^2$ copies can be ignored by the adversary, as such copies contribute at most $O(n f(n)^2 N^2)$ total weight to the algorithm’s solution, which is negligible compared to the inflated optimal solution.

For each type containing at least $f(n)N^2$ copies, the adversary can partition these copies into groups of exactly $N$, discarding at most $O(N)$ leftover copies per type. The total contribution of all discarded copies remains bounded by $O(n f(n)^2 N^2)$.

As a result, the adversary can extract $N$ identical copies of the instance on which the algorithm behaves identically. Since the value of the optimal solution under this amplification grows to $\Omega(n f(n)^2 N^3)$, the ignored portions of the instance affect the algorithm’s competitive ratio by only an $o(1)$ term. This justifies the assumption that the adversary may freely create identical copies of the instance in our hardness construction.

\bibliographystyle{alpha}
\bibliography{ref}

\end{document}